\documentclass[12pt,reqno]{amsart}
\usepackage[hypertex]{hyperref}
\usepackage[final]{graphicx}
\usepackage{amsfonts}
\usepackage{pdfsync}
\usepackage{amsmath}
\usepackage{amssymb}
\usepackage{amsthm}
\newcommand{\dd}{\mathrm{d}}

\newtheorem{theorem}{Theorem}[section]
\newtheorem{lemma}[theorem]{Lemma}

\newtheorem{remark}[theorem]{Remark}
\newtheorem*{remarks}{Remarks}

\newcommand{\R}{{\mathord{\mathbb R}}}

\begin{document}
\title[A New Estimate on the Two-Dimensional Indirect Coulomb Energy]{\textbf{A New Estimate on the Two-Dimensional Indirect Coulomb Energy}}

\author[Benguria]{Rafael D. Benguria$^1$}

\author[Gallegos]{Pablo Gallego$^2$}

\author[Tu\v{s}ek]{Mat\v{e}j Tu\v{s}ek$^3$}

\address{$^1$ Departmento de F\'\i sica, P. Universidad Cat\' olica de Chile,}
\email{\href{mailto: rbenguri@fis.puc.cl}{ rbenguri@fis.puc.cl}}

\address{$^2$ Departmento de F\'\i sica, P. Universidad Cat\' olica de Chile, }
\email{\href{mailto: pigalleg@uc.cl}{pigalleg@uc.cl}}

\address{$^2$ Departmento de F\'\i sica, P. Universidad Cat\' olica de Chile, }
\email{\href{mailto: tmatej@gmail.com}{tmatej@gmail.com}}

% \date{\today}
\maketitle

\smallskip
\begin{abstract}
We prove  a new lower bound on the indirect Coulomb energy in two dimensional quantum mechanics in terms of the single particle density of the system. The new universal lower bound is an alternative to the Lieb--Solovej--Yngvason  bound with a smaller constant, $C=(4/3)^{3/2}\sqrt{5 \pi -1} \approx 5.90  < C_{LSY}=
192\sqrt{2 \pi} \approx 481.27$, which also  involves an additive gradient energy term  of the single particle density. 
\end{abstract}

\section{Introduction}

Since the beginning of Quantum Mechanics there has been a wide interest in estimating various energy terms of a system of electrons in terms of the single particle density $\rho_{\psi}(x)$.  Given that the expectation value of the Coulomb attraction of the electrons by the nuclei can be expressed in closed form in terms of $\rho_{\psi}(x)$, the interest focuses on estimating the expectation value of the kinetic energy of the system of electrons and on  the expectation value of the Coulomb repulsion between the electrons.  
Here, we will be concerned with the latest. The most natural approximation to the expectation value of the Coulomb repulsion between the electrons is given by
\begin{equation}
D(\rho,\rho)=\frac{1}{2} \int \rho(x) \frac{1}{|x-y|} \rho(y) \, \dd x \, \dd y, 
\end{equation}
which is usually called the {\it direct term}. 
The remainder, i.e., the difference between the expectation value of the electronic repulsion and $D(\rho,\rho)$, say $E$, is called the {\it indirect term}. 
In 1930, Dirac \cite{Di30} gave the first approximation to the indirect Coulomb energy in terms of the single particle density. Using an argument with  plane waves, he approximated $E$ by 
\begin{equation}
E \approx -c_D e^{2/3} \int \rho^{4/3} \, dx,
\label{eq:dirac}
\end{equation}
where $c_D=(3/4)(3/\pi)^{1/3} \approx  0.7386$ (see, e.g., \cite{Mo06}, p. 299). Here $e$ denotes the absolute value of the charge of the electron. 
The first rigorous lower bound for $E$ was obtained by E.H. Lieb in 1979 \cite{Li79}, using the Hardy--Littlewood Maximal Function \cite{StWe71}. There he found that,
$E \geq -8.52 e^{2/3} \int \rho^{4/3} \, dx$.  The constant 8.52 was substantially improved by E.H. Lieb and S. Oxford in 1981 \cite{LiOx81}, who proved the bound
\begin{equation}
E  \ge  -C e^{2/3} \int \rho^{4/3} \, dx,
\label{eq:LO}
\end{equation}
with  $C = c_{LO}=1.68$.  The best value for $C$ is unknown, but Lieb and Oxford \cite{LiOx81} 
proved that it is larger or equal than $1.234$. The Lieb--Oxford value was later improved to $1.636$ by 
Chan and Handy, in 1999 \cite{ChHa99}.  It is this last constant, as far as we know, that is the smallest value for $C$ that has been found to this day. During the last thirty years, after the work of Lieb and Oxford \cite{LiOx81},  there has been a special interest in quantum chemistry in constructing corrections to the Lieb--Oxford term involving the gradient of the single particle density. This interest arises with the expectation that states with a relatively small kinetic energy have a smaller indirect part (see, e.g., \cite{LePe93,PeBuEr96,VeMeTr09} and references therein). Recently, Benguria, Bley, and Loss obtained an alternative to (\ref{eq:LO}), which has a lower constant (close to $1.45$) to the expense of adding a gradient term (see Theorem 1.1 in  \cite{BeBlLo11}).

After the work of Lieb and Oxford \cite{LiOx81} many people have considered bounds on the indirect Coulomb energy in lower dimensions (in particular see, e.g., \cite{HaSe01} for the one dimensional case,  \cite{LiSoYn95}, \cite{NaPoSo11}, \cite{RaPiCaPr09} and \cite{RaSeGo11} for the two dimensional case, which is important for the study of quantum dots). In this manuscript we give an alternative to the Lieb--Solovej--Yngvason bound \cite{LiSoYn95}, with a constant  much closer to the numerical values proposed in \cite{RaSeGo11} (see also the references therein) 
to the expense of adding a gradient term. In some sense, the result proven here is the analog of the three dimensional  result proven in  \cite{BeBlLo11} for two dimensional systems.

\bigskip

Our main result is the following theorem:

\begin{theorem}[Estimate on the indirect Coulomb energy in two dimensions]\label{thm:LO}
Let $\psi\in L^{2}(\R^{2N})$ be normalized to one and symmetric (or antisymmetric) in all its variables. Define
$$\rho_{\psi}(x)=N\int_{\R^{2(N-1)}}|\psi|^{2}(x,x_{2},\ldots,x_{N})~\dd x_{2}\ldots\dd x_{N}.$$
Then, for all $\epsilon>0$, 
\begin{equation}\label{eq:ind_en_est}
 E(\psi)\equiv\langle\psi,\sum_{i<j}^{N}|x_{i}-x_{j}|^{-1}\psi\rangle-D(\rho_{\psi},\rho_{\psi})\geq -(1+\epsilon)\beta\int_{\R^{2}}\rho_{\psi}^{3/2} \, \dd x-\frac{4}{\beta\epsilon}\int_{\R^{2}}|\nabla\rho_{\psi}^{1/4}|^{2} \, \dd x
\end{equation}
with
\begin{equation*}
\beta=\left(\frac{4}{3}\right)^{3/2}\sqrt{5\pi-1}\simeq 5.9045. 
\end{equation*}
\end{theorem}

\begin{remarks}

\bigskip
\noindent
i) Our constant $\beta \simeq 5.9045$ is substantially lower than the  constant $C_{LSY} \simeq 481.27$ found in \cite{LiSoYn95} (see equation (5.24) of lemma 5.3 in  \cite{LiSoYn95}), which is the best bound to date. 

\bigskip
\noindent
ii) The constant $\beta$  is close to the numerical values (i.e., $\simeq 1.95$)  of \cite{RaPiCaPr09} (and references therein), but is not sharp.

\end{remarks} 

Our proof relies on a stability result for an auxiliary molecular quantum system in two dimensions (which is proven in Section 2) and an observation of Lieb and Thirring \cite{LiTh75}. The proof of the main theorem is given in Section 3.

\section{A stability result for an auxiliary molecular system in two dimensions}

A key role in our proof of the Lieb--Oxford type bound in two dimensions will be played by a stability result on an auxiliary molecular system in two dimensions. This molecular system may be viewed as the two dimensional version of the zero mass limit of the relativistic Thomas--Fermi--Weizs\"acker  energy functional studied in \cite{BeLoSi07} (which corresponds to the zero mass limit of the model introduced in \cite{En87,EnDr87,EnDr88}; the stability properties of the corresponding atomic system were studied also in \cite{BePe02}). Thus, let us consider the energy functional
\begin{equation} 
\xi(\rho)= a^2 \int_{\R^2} (\nabla \rho^{1/4})^2 \, \dd x + b^2 \int_{\R^2} \rho^{3/2} \, \dd x - \int_{\R^2} V(x) \rho (x) \, \dd x +D(\rho,\rho) +U,
\label{eq:s1}
\end{equation}
where the potential $V$ is given by
\begin{equation}
V(x) = \sum_{i=1}^K \frac{z}{|x-R_i|},
\label{eq:s2}
\end{equation} 
which may be viewed as a Coulomb--like potential generated by $K$ point particles (nuclei) of (equal)  charge $z>0$, located at $R_i \in \R^2$ (with $i=1,\dots K$). 
Here, the function $\rho(x) \ge 0$ is the electronic density of a system of $N$ electrons, and 
\begin{equation}
D(\rho,\rho) = \frac{1}{2} \int_{\R^2 \times \R^2} \rho(x) \frac{1}{|x-y|} \rho(y) \, \dd x \, \dd y
\label{eq:s3}
\end{equation}
is the electronic repulsion energy.  Finally, 
\begin{equation}
U = \sum_{1 \le i < j \le K} \frac{z^2}{|R_i-R_j|}.
\label{eq:s4}
\end{equation}
The powers of the first two terms in (\ref{eq:s1}), i.e., 
\begin{equation}
{T}(\rho) =  a^2 \int_{\R^2} (\nabla \rho^{1/4})^2 \, \dd x + b^2 \int_{\R^2} \rho^{3/2} \, \dd x, 
\label{eq:s5}
\end{equation}
are such that $T(\rho_{\alpha}) = \alpha T(\rho)$, where $\rho_{\alpha}(x) = \alpha^2 \rho(\alpha x)$  (with $\alpha>0$) is such that $\int_{\R^2} \rho_{\alpha} \, \dd x = 
\int_{\R^2} \rho(x) \, \dd x$. In other words, the kinetic energy of the electrons scales like one over a length, i.e., in the same way as the potential energy.  Then, as
usual in this situation, the values of the coupling constant (i.e., the values of the nuclear charge) will be crucial to ensure stability of the system. 
Our main result in this section is the following stability theorem, which is the two dimensional  analog of Theorem 1.2 in \cite{BeLoSi07}.

\begin{theorem}\label{theo:main}
 For any  $a, b>0$, and $R_i \in \R^2$, $i=1, \dots, K$, and for all $\rho \ge 0$ (with $\rho \in L^{3/2} (\R^2)$ and $\nabla \rho^{1/4} \in L^{2} (\R^2)$), we have that 
\begin{equation} 
\xi(\rho) \equiv a^2 \int_{\R^2} (\nabla \rho^{1/4})^2 \, \dd x + b^2 \int_{\R^2} \rho^{3/2} \, \dd x - \int_{\R^2} V(x) \rho (x) \, \dd x +D(\rho,\rho) +U \ge 0, 
\label{eq:s6}
\end{equation}
 where $V$, $D$, and $U$ are defined  by (\ref{eq:s2}),  (\ref{eq:s3}), and  (\ref{eq:s4}), respectively, provided, 
 \begin{equation}
 0 \le z \le z_c(a,b) \equiv \frac{a \, b}{2} \sqrt{1-\sigma}. 
 \label{eq:s7}
 \end{equation}
 Here $0 < \sigma <1$ is the only positive root of the quartic equation 
 \begin{equation}
 \frac{\sigma^{2}}{\sqrt{1-\sigma}}=\frac{32(5\pi-1)}{27}\frac{a}{b^{3}}
 \label{eq:s8}
 \end{equation}
 on the interval $(0,1)$.
 \end{theorem}

In the rest of this section we will give the proof of this theorem, which is similar to the proof of Theorem 1.2 in \cite{BeLoSi07}. Notice that the upper limit $z_c(a,b)$ on $z$ to insure stability is not sharp; in other words, there could still be values of $z$ above our $z_c$ for which $\xi(\rho) \ge 0$. We start with an appropriate {\it Coulomb uncertainty principle}.

\begin{theorem}\label{theo:uncer_princ}
For any smooth function $f$ on the closed disk $D_{R}$, of radius $R$,  and for all  $a,\, b\in\R$, we have
$$a^{2}\int_{D_{R}}|\nabla f(x)|^{2}\, \dd x+b^{2}\int_{D_{R}}f(x)^{6}\, \dd x\geq ab\int_{D_{R}}\left(\frac{1}{2|x|}-\frac{1}{R}\right)f(x)^{4}\, \dd x.$$
\end{theorem}

\noindent
To prove the theorem one only imitates the proof of Theorem 2.1 in \cite{BeLoSi07} that deals with the three-dimensional case. We start with the following preliminary result which may be of  independent interest.

\begin{lemma}
Let $u=u(|x|)$ be a sufficiently smooth real function (so that all the terms in (\ref{eq:uncer_princ_lem}) are finite) defined on the interval $[0,R]$, such that $u(R)=0$. Then the following uncertainty principle holds
\begin{equation}\label{eq:uncer_princ_lem}
\left|\int_{D_{R}}[2u(|x|)+|x|u'(|x|)]f(x)^{4}\, \dd x\right|\leq 4\left(\int_{D_{R}}|\nabla f(x)|^{2}\, \dd x\right)^{1/2}\left(\int_{D_{R}}u(|x|)^{2}|x|^{2}f(x)^{6}\, \dd x\right)^{1/2}.
\end{equation}
In (\ref{eq:uncer_princ_lem}) there is equality if and only if
\begin{equation}
f(x)^{2}=\frac{1}{\lambda\int_{0}^{|x|}su(s)\dd s+C},
\label{eq:function}
\end{equation}
for some constants $C$ and $\lambda$.
\end{lemma}
\begin{proof}
Set $g_{j}(x)=u(|x|)x_{j}$. Then we have,
\begin{equation*}
 \begin{split}
  &\int_{D_{R}}[2u(|x|)+|x|u'(|x|)]f(x)^{4}\, \dd x=\sum_{j=1}^{2}\int_{D_{R}}[\partial_{j}g_{j}(x)]f(x)^{4}\, \dd x= \\
  &=\sum_{j}\int_{D_{R}}f(x)\partial_{j}[g_{j}(x)f(x)^{3}]\, \dd x-3\sum_{j}\int_{D_{R}}f(x)^{3}g_{j}(x)\partial_{j}f(x)\, \dd x=\\
  &=-4\int_{D_{R}}\langle \nabla f(x),\, x\rangle u(|x|)f(x)^{3}\, \dd x.
 \end{split}
\end{equation*}
In the last equality we integrated by parts and made use of the fact that $u$ vanishes on the boundary $\partial D_{R}$. Next, the Schwarz inequality implies 
$$\left|\int_{D_{R}}[2u(|x|)+|x|u'(|x|)]f(x)^{4}\, \dd x\right|\leq 4\left(\int_{D_{R}}|\nabla f(x)|^{2}\, \dd x\right)^{1/2}\left(\int_{D_{R}}u(|x|)^{2}|x|^{2}f(x)^{6}\, \dd x\right)^{1/2}.$$
In the last expression, equality is obtained  if and only if
$$\partial_{j}f(x)=-\frac{\lambda}{2}x_{j}u(|x|)f(x)^{3},$$
which after an integration yields the function given by (\ref{eq:function}) above.
\end{proof}

\textit{Proof of Theorem \ref{theo:uncer_princ}.} Choosing $u(r)=r^{-1}-R^{-1}$ in  (\ref{eq:uncer_princ_lem}), we conclude that
\begin{equation*}
\begin{split}
&\frac{ab}{2}\left|\int_{D_{R}}\left(\frac{1}{|x|}-\frac{2}{R}\right)f(x)^{4}\dd x\right|\leq 2ab \left(\int_{D_{R}}|\nabla f(x)|^{2}\dd x\right)^{1/2}\left(\int_{D_{R}}f(x)^{6}\dd x\right)^{1/2}\leq\\
&\leq a^{2}\int_{D_{R}}|\nabla f(x)|^{2}\dd x+b^{2}\int_{D_{R}}f(x)^{6}\dd x.
\end{split}
\end{equation*}  
\hfill$\square$

To prove our main result of this section, i.e., Theorem \ref{theo:main}, we will also need the following auxiliary lemma.
\begin{lemma}\label{lem:halfplane_int}
Let $D_{L}(x_{0})=\{ x\in\R^{2} \bigm| \ |x-x_{0}|<L\}$ and $H$ be a half plane such that $\mathrm{dist}(x_{0},\partial H)=L$ and $x_{0}\in H$. Then 
$$\int\limits_{H\setminus D_{L}(x_{0})}\frac{1}{|x-x_{0}|^{3}} \, \dd x=\frac{2(\pi-1)}{L}.$$ 
\end{lemma}
\begin{proof}
Let us shift the origin of the coordinates to $x_{0}$ and choose the $x$ Cartesian axes parallel to $\partial H$. Then in the respective polar coordinates $(\varrho,\varphi)$,
$$\int\limits_{H\setminus D_{L}(x_{0})}\frac{1}{|x-x_{0}|^{3}}\, \dd x=2 \int_{-\pi/2}^{0}\int_{L}^{\frac{L}{\cos{\varphi}}}\frac{1}{\varrho^{2}}\, \dd\varrho \, \dd\varphi+
2 \int_{0}^{\pi/2}\int_{L}^{\infty}\frac{1}{\varrho^{2}}\, \dd\varrho \, \dd\varphi$$
from which the assertion of the lemma follows by a straightforward integration. 
\end{proof}

In the sequel we need some notation. We introduce the nearest neighbor, or Voronoi,
cells \cite{Vo07} (see also the review \cite{LiSe09}), $\{\Gamma_j\}_{j=1}^K$, defined by
\begin{equation}
\Gamma_j=\{x \bigm| |x-R_j| \le |x-R_k|\}.
\end{equation}
The boundary of $\Gamma_j$, $\partial \Gamma_j$, consists of a finite number of lines. We also define the distance\begin{equation}
D_j={\rm dist}(R_j,\partial \Gamma_j) = \frac{1}{2} \min \{|R_k-R_j| \bigm| k \neq j \}.
\end{equation}
Finally, we denote by $B_j$ the disk of radius $D_j$ centered at $R_j$, $j=1,\dots,K$.

One of the key ingredients we need in the sequel is the two dimensional version, (see e.g., \cite{LiSoYn94}), of an electrostatic inequality of
Lieb and Yau \cite{LiYa88a,LiYa88b}.  
Define the piecewise function $\Phi(x)$ on $\R^2$ with the aid of the Voronoi cells mentioned above. 
In the cell $\Gamma_j$,  $\Phi(x)$  equals the electrostatic potential generated by
all the nuclei except for the nucleus situated in $\Gamma_j$ itself, i.e., for $x \in \Gamma_j$,
\begin{equation}
\Phi(x) = \sum_{\substack{i=1\\ i \neq j}}^K \frac{z}{|x-R_i|}.
\end{equation}
Then, one has (see, e.g., \cite{LiSoYn94}), 
\begin{equation}
D(\rho,\rho) - \int_{\R^2} \Phi(x) \rho(x) \, \dd x + U \ge \frac{z^2}{8} \sum_{j=1}^K \frac{1}{D_j}.
\label{eq:BEI}
\end{equation}
This follows at once from the standard (three dimensional)  Lieb--Yau electrostatic inequality \cite{LiYa88a,LiYa88b} by taking a Borel measure supported on a two dimensional 
plane, with density $\rho(x)$. 

With the help of the Coulomb uncertainty principle and the two dimensional electrostatic inequality (\ref{eq:BEI}), we are ready to prove the following estimate. 

\begin{lemma}\label{lem:main_est}
 For any  $\rho \in L^{3/2} (\R^2)$ such that $\nabla \rho^{1/4} \in L^{2} (\R^2)$; for all $b_1 > 0$, and $b_2>0$ such that $b_1^2+b_2^2 = b^2$, and $z\le a \, b_2/2$, we have
 \begin{equation}\label{eq:main_est}
 \xi(\rho)\geq \sum_{j=1}^{K}\frac{1}{D_{j}}\left[\frac{z^{2}}{8}-\frac{4}{27b_{1}^{4}}\left(2z^{3}(\pi-1)+\pi a^{3}b_{2}^{3}\right)\right].
\end{equation}
\end{lemma}
\begin{proof}
Setting $f(x)^4=\rho(x)$, splitting $\R^2$ as the disjoint union of the Voronoi cells $\Gamma_j$, using
Theorem \ref{theo:uncer_princ} in each disk $B_j$, and discarding the kinetic energy terms (which are positive) in the complements $\Gamma_j \setminus B_j$ 
we conclude that
\begin{equation}\label{eq:initial_est}
 \xi(\rho)\geq b_{1}^{2}\int_{\R^{2}}\rho^{3/2}\dd x-\int_{\R^{2}}V\rho~\dd x+ab_{2}\sum_{j=1}^{K}\int_{B_{j}}\left(\frac{1}{2|x-R_{j}|}-\frac{1}{D_{j}}\right)\rho(x)\dd x+
 D(\rho,\rho)+U.
\end{equation}
It is convenient to define the piecewise function $W(x)$ as
\begin{equation}
 W(x)=\begin{cases}
        \Phi(x)+\dfrac{z}{|x-R_{j}|}=V(x)&\text{if }x\in\Gamma_{j}\setminus B_{j} \\
	\Phi(x)+\dfrac{ab_{2}}{D_{j}}&\text{if }x\in B_{j},
       \end{cases}
\label{eq:defW}
\end{equation}
Provided $z \le a \, b_2/2$ (which we assume from here on), we can estimate from below the sum of the second and third integrals in (\ref{eq:initial_est}) in terms of $W(x)$ as follows, 
\begin{equation*}
 \begin{split}
 & ab_{2}\sum_{j=1}^{K}\int_{B_{j}}\left(\frac{1}{2|x-R_{j}|}-\frac{1}{D_{j}}\right)\rho(x) \, \dd x-\int_{\R^{2}}V\rho \, \dd x\\
 &=ab_{2}\sum_{j=1}^{K}\int_{B_{j}}\left(\frac{1}{2|x-R_{j}|}-\frac{1}{D_{j}}\right)\rho(x)\, \dd x-z\sum_{i,j=1}^{K}\int_{\Gamma_{j}\setminus B_{j}}\frac{\rho(x)}{|x-R_{i}|}\, \dd x\\
 &-z \sum_{\stackrel{i,j=1}{i\neq j}}^{K}\int_{B_{j}}\frac{\rho(x)}{|x-R_{i}|}\, \dd x-z\sum_{j=1}^{K}\int_{B_{j}}\frac{\rho(x)}{|x-R_{j}|}\, \dd x\\
 &=-\int_{\R^{2}}W(x)\rho(x) \,\dd x+\sum_{j=1}^{K}\int_{B_{j}}\left(\frac{ab_{2}}{2}-z\right)\frac{\rho(x)}{|x-R_{j}|}\, \dd x\geq -\int_{\R^{2}}W(x)\rho(x) \, \dd x.
 \end{split}
\end{equation*}
Thus, we can write
\begin{equation}\label{eq:xi_decomp}
\xi(\rho)\geq \xi_{1}(\rho)+\xi_{2}(\rho),
\end{equation}
with
\begin{align*}
 &\xi_{1}(\rho)=b_{1}^{2}\int_{\R^{2}}\rho^{3/2} \, \dd x-\int_{\R^{2}}(W-\Phi)(x)\rho(x) \, \dd x \qquad \mbox{and,}\\
 &\xi_{2}(\rho)=D(\rho,\rho)-\int_{\R^{2}}\Phi(x)\rho(x) \, \dd x+U.
\end{align*}
From the definition of $\xi_1(\rho)$, it is clear that $\xi_1(\rho) \ge \xi_1(\hat \rho)$, where ${\hat \rho}(x) = 4(W(x)-\Phi(x))_+^2/(9b_1^4)$, where as usual $u_+=\max(u,0)$. 
Hence, 
\begin{equation*}
 \xi_{1}(\rho)\geq -\frac{4}{27b_{1}^{4}}\int_{\R^{2}}(W-\Phi)_{+}^{3} \, \dd x=-\frac{4}{27b_{1}^{4}}\sum_{j=1}^{K}\left(\int_{\Gamma_{j}\setminus B_{j}}\frac{z^{3}}{|x-R_{j}|^{3}} \, \dd x+\int_{B_{j}}\left(\frac{ab_{2}}{D_{j}}\right)^{3}\, \dd x\right),
\end{equation*}
where the last equality follows from the definition (\ref{eq:defW}) of $W$.
As every $\Gamma_{j}$ is contained in a half-plane, we may estimate the first integral above with the help of Lemma \ref{lem:halfplane_int}. This way we get
\begin{equation}\label{eq:est_xi_1}
\xi_{1}(\varrho)\geq -\frac{4}{27b_{1}^{4}}\left[2z^{3}(\pi-1)+\pi a^{3}b_{2}^{3}\right]\sum_{j=1}^{K}\frac{1}{D_{j}}.
\end{equation}
The lower bound for $\xi_{2}(\rho)$ follows at once from (\ref{eq:BEI}), i.e., 
\begin{equation}\label{eq:est_xi_2}
 \xi_{2}(\varrho)\geq\frac{z^{2}}{8}\sum_{j=1}^{N}\frac{1}{D_{j}}.
\end{equation}
Putting (\ref{eq:xi_decomp}), (\ref{eq:est_xi_1}), and (\ref{eq:est_xi_2}) together the assertion of the lemma immediately follows. 
\end{proof}

We end this section with the proof of Theorem \ref{theo:main}. 
\begin{proof}[Proof of Theorem \ref{theo:main}]
Let $M(z)$ stand for the term inside the square brackets on the right side of (\ref{eq:main_est}). With $z_c$ and $\sigma$ defined by  by (\ref{eq:s7}) and (\ref{eq:s8}) respectively, 
set $p=z/z_c$, and $b_2=  p \, b \sqrt{1-\sigma}$. Hence, $b_1^2=b^2-b_2^2=b^2(1-p^2+ p^2 \sigma)$. Replacing the expressions of $b_1$ and $b_2$ in 
the expression for $M(z)$ we get, 
\begin{equation}
M(z) = \frac{p^2}{32} a^2\, b^2(1-\sigma) \left[1- \frac{32 \, a}{27 b^3} h(p) \sqrt{1-\sigma} (5 \pi -1) \right],
\label{eq:M1}
\end{equation}
where
$$
h(p) \equiv \frac{p}{(1-p^2+ p^2 \sigma)^2}.
$$
Here, both $p$ and $\sigma$ belong to the interval $[0,1]$. It is simple to see that $h(p)$ is strictly increasing in the interval $[0,1]$, thus $h(p) \le h(1)=1/\sigma^2$. 
Using this last inequality in (\ref{eq:M1}) together with the definition of $\sigma$, i.e.,  equation (\ref{eq:s8}), we conclude that 
$$
M(z) \ge 0, 
$$
for all $z \le z_c$. 
\end{proof}

\section{Proof of Theorem \ref{thm:LO}}

In this Section we give the proof of the main result of this paper, namely Theorem \ref{thm:LO}. We use an idea introduced by Lieb and Thirring in 1975 in their proof of the stability of matter \cite{LiTh75} (see also the review article \cite{Li76} and the recent monograph \cite{LiSe09}).

\begin{proof}[Proof of Theorem \ref{thm:LO}]
Consider the inequality (\ref{eq:s6}), with $K=N$ (where $N$ is the number of electrons in our original system), $z=1$ (i.e., the charge of the electrons), and $R_i=x_i$ (for all $i=1, \dots, N$). With this choice, according to (\ref{eq:s7}), the inequality  (\ref{eq:s6}) is valid as long  as $a$ and $b$ satisfy the constraint, 
\begin{equation}
2 \le a b \sqrt{1-\sigma},
\label{eq:3.1}
\end{equation}
with $\sigma \in (0,1)$ the solution of 
\begin{equation}
\frac{\sigma^2}{\sqrt{1-\sigma}} = \frac{\beta^2}{2} \frac{a}{b^3}
\label{eq:3.2}
\end{equation}
where
\begin{equation}
\beta = \left(\frac{4}{3}\right)^{3/2} \sqrt{5 \pi -1} \simeq 5.9045. 
\label{eq:3.3}
\end{equation}
Then take any normalized wavefunction $\psi(x_1,x_2, \dots, x_N)$, and multiply (\ref{eq:s6})  by $| \psi(x_1, \dots, x_N)|^2$ and integrate  over all the electronic configurations, i.e., on $\R^{2N}$. Moreover, take $\rho=\rho_{\psi}(x)$. We get at once, 
\begin{equation}
 E(\psi)\equiv\langle\psi,\sum_{i<j}^{N}|x_{i}-x_{j}|^{-1}\psi\rangle-D(\rho_{\psi},\rho_{\psi})
 \geq 
 -b^2 \int_{\R^{2}}\rho_{\psi}^{3/2} \, \dd x  - a^2 \int_{\R^{2}}|\nabla\rho_{\psi}^{1/4}|^{2} \, \dd x,
\label{eq:3.4}
\end{equation}
provided $a$ and $b$ satisfy (\ref{eq:3.1}) and (\ref{eq:3.2}) above. 
Thinking of $\sigma \in (0,1)$ as a free parameter, and $a, b$ satisfying (\ref{eq:3.1}) and (\ref{eq:3.2}), and writing $\varepsilon=(1-\sigma)/\sigma$ we get at once from 
(\ref{eq:3.1}) and (\ref{eq:3.2}) that 
$$
b^2 \ge (1+\epsilon) \beta, 
$$
for any $\epsilon>0$.
The theorem then follows by choosing the minimum value of $b^2$, i.e., $b^2=(1+\epsilon)\beta$, hence $a^2=4/(\beta \epsilon)$. 
\end{proof}

\begin{remark}
In general the two integral terms in (\ref{eq:ind_en_est}) are not comparable. If one takes a very rugged $\rho$, normalized to $N$, the gradient term may be very large while the other term can remain small. However, if one takes a smooth $\rho$, the gradient term can be very small as we illustrate in the example below. 
Let us denote
$$
L(\rho)=\int_{\R^2} \rho(x)^{3/2} \, \dd x
$$
and 
$$
G(\rho)=\int_{\R^2} (\nabla \rho(x)^{1/4})^2 \, \dd x.
$$
We will evaluate them for the normal distribution
 $$\rho(|x|)=C\mathrm{e}^{-A|x|^{2}}$$
where $C,\, A>0$. Some straightforward integration yields
 $$L=C^{\frac{3}{2}}\frac{2\pi}{3A},\quad G=C^{\frac{1}{2}}\pi.$$ 
With $C=NA/\pi$,
$$\int_{\R^{2}}\rho(|x|) \, \dd x=N,$$
and we have
$$\frac{G}{L}=\frac{3\pi}{2N},$$
i.e., in the ``large number of particles'' limit, the $G$ term becomes negligible.
\end{remark}

\bigskip

\section*{Acknowledgments}

This work has been supported by the Iniciativa Cient'fica Milenio, ICM (CHILE) project P07--027-F.   The work of RB has also been supported by FONDECYT (Chile) Project 1100679. The work of MT has also been partially supported by the grant 201/09/0811 of the Czech Science Foundation.

We would like to thank Jan Philip Solovej for useful remarks.

\end{document}